\documentclass[a4paper,oneside]{article}
\usepackage[margin=1in]{geometry}
\usepackage{todonotes}

\usepackage{amsthm,amssymb,amsmath,amsfonts}
\usepackage{tabularx,lipsum,environ,hyperref}
\usepackage{float}
\makeatletter
\newcommand{\problemtitle}[1]{\gdef\@probtitle{#1}}
\newcommand{\probleminput}[1]{\gdef\@probgiven{#1}}
\newcommand{\problemquestion}[1]{\gdef\@probfind{#1}}
\newcounter{problem}
\NewEnviron{problem}{
 \refstepcounter{problem}
  \problemtitle{}\probleminput{}\problemquestion{}
  \BODY
  \par\addvspace{.5\baselineskip}
  \noindent
  \begin{tabularx}{\textwidth}{@{\hspace{\parindent}} l X c}
    \multicolumn{2}{@{\hspace{\parindent}}l}{\@probtitle} \\
    \textbf{Given:} & \@probgiven \\
    \textbf{Find:} & \@probfind
  \end{tabularx}
  \par\addvspace{.5\baselineskip}
}
\makeatother

\makeatletter
\renewenvironment{proof}[1][\proofname]{%
  \par
  \pushQED{\qed}%
  \normalfont
  \topsep6\p@\@plus6\p@\relax
  \trivlist
  \item[\hskip\labelsep
        \itshape
    #1\@addpunct{.}]\mbox{}\newline
}{%
  \popQED
  \endtrivlist\@endpefalse
}
\makeatother

\newtheorem{theorem}{Theorem}[section]

\newtheorem{lemma}{Lemma}[section]

\newtheorem{thm}[theorem]{Theorem}

\newtheorem{prop}[theorem]{Proposition}

\def\rtwo{\mathbb{R}^2}
\def\planarized{^{\mathsf P}}

\DeclareMathOperator{\opt}{OPT}

\DeclareMathOperator{\etsp}{ETSP}

\DeclareMathOperator{\mst}{MST}

\DeclareMathOperator*{\argmin}{argmin}
\DeclareMathOperator{\copies}{copies}

\DeclareMathSymbol{\shortmathminus}{\mathbin}{AMSa}{"39}
\newcommand{\shortminus}{\mathbin{\!\shortmathminus\!}}

\newcommand{\setcond}[2]{\{ #1 \enspace | \enspace #2\}}
\newcommand{\set}[1]{\{ #1 \}}

\newcommand{\sequence}[1]{\langle #1 \rangle}

\usepackage{algorithm}
\usepackage[noend]{algpseudocode}

\usepackage{bm}
\usepackage{verbatim}
\usepackage{paralist}
\usepackage{tikz}
\usepackage{stmaryrd}
\usepackage{paralist}
\usepackage{subcaption}

\usepackage{authblk}

\title{%
  Polynomial-time $(k+\epsilon)$-approximation for $k$-coloured Non-crossing Euclidean TSP%
  \thanks{This research was supported by the Deutsche Forschungsgemeinschaft via the cluster of excellence PhenoRob.}
}

\author{%
  Daniel Bauer\thanks{University of Bonn, Email: \texttt{bauer@igg.uni-bonn.de}} \quad %
  Jan-Henrik Haunert\thanks{University of Bonn, Email: \texttt{haunert@igg.uni-bonn.de}}%
}
\begin{document}
\maketitle

\begin{abstract}
 Given a $k$-coloured point set $P\subseteq \rtwo$, the $k$-coloured Non-crossing Euclidean Travelling Salesperson Problem (short $k$-ETSP) asks for $k$ non-crossing closed curves, where one curve spans one corresponding colour class, such that the curves are pairwise non-crossing  and the sum of their Euclidean lengths is minimised. This problem is NP-hard as $1$-ETSP is the standard Euclidean Travelling Salesperson Problem. We present a polynomial-time $(k+\epsilon)$-approximation for $k$-ETSP.
\end{abstract}

\section{Introduction and related work}
Connecting within groups without causing interference to other group-connections is a widespread task. Major examples are found in cartography and very-large-scale integration (VLSI).
The latter deals with the design of micro-chip layouts, which has to follow many physical constraints, e.g., wires that carry different signals must be placed without obstruction. 

In cartography, related problems arise in the context of maps that show points of different categories, such as restaurants, museums, and bars. Current research deals with connecting points of equal category, e.g., by closed regions or line segments, to enhance the visual perception of the map, whilst minimising interference, e.g., overlap between regions or crossings between line segments \cite{efrat2014mapsets,van2024simplesets}. 

In this work, we focus on finding closed curves connecting points of equal colour whilst avoiding any crossings, namely the $k$-coloured Non-crossing Euclidean Travelling Salesperson Problem (short $k$-ETSP).

The Travelling Salesperson Problem (TSP) is an old fundamental problem in combinatorial optimisation. Some of its first mentions date back to the 17th century. Over time, TSP gained much theoretical and practical attention, developing it to a household name in combinatorial optimisation and computer science.
We refer to the standard literature on combinatorial optimisation \cite{korte2011traveling}, which was used as a source for the following paragraphs.

The TSP can be formulated as: Given a complete graph with non-negative edge-weights, find a minimum-weight cycle visiting every vertex in $G$ exactly once. Analytical results on TSP are predominantly negative; solving it is NP-hard, as is approximating it within a constant factor. 

A metric graph is a complete non-negative weighted graph such that any path $P$ from $v$ to $w$ is at least as long as the edge connecting $v$ and $w$. 
Metric TSP, in which only  metric graphs are considered, is strongly NP-hard and APX-hard. But it admits a $\frac{3}{2}$-approximation introduced by Christofides \cite{VANBEVERN2020118}. The algorithm computes for a given metric graph $G$ a minimum spanning tree $\mst(G)$ and a minimum-weight perfect matching $M$ of the odd vertices of $\mst(G)$ with respect to $G$.
The union of $\mst(G)$ and $M$ yields an Eulerian graph $G'$ of total weight less than $\frac{3}{2}$ times an optimal TSP tour in $G$. An Eulerian circuit of $G'$ is then used to obtain a spanning cycle of length less than the total length of $T$ and $J$.

The Euclidean Travelling Salesperson Problem (short ETSP) is further restricted to metric graphs
for which their vertices form a planar point set and each edge weight is the Euclidean distance of its endpoints. Naturally, ETSP found appeal in many practical applications.

Analogously to TSP and ETSP, other famous combinatorial optimisation problems are considered under restriction to Euclidean space. 
Given a point set $P$, the Euclidean Geometric Steiner Tree Problem (short ESMT) asks for a set of line segments that form a tree which spans $P$ and is of minimum length.
Arora~\cite{arora1998polynomial} famously introduced a $(1+\epsilon)$-approximation algorithm inter alia for ETSP and ESMT, therefore proving alongside that ETSP is not APX-hard in contrast to Metric TSP. Arora's algorithm subdivides the plane into multiple levels of grids. From there on, only certain candidate solutions are considered that intersect the grids in special points called portals. Via dynamic programming, a minimum length tour restricted to such portals can be computed in polynomial time, which is proven to be a $(1+\epsilon)$-approximation of an optimal solution of ETSP.    

Recent research considers Euclidean combinatorial optimisation problems where the aim is to compute multiple solution pieces that must be non-crossing to one another.
A current example is the $k$-coloured Non-crossing Euclidean Steiner Forest Problem (short $k$-CESF): Given a $k$-coloured point set $P$, find for each colour a Steiner tree of the corresponding points such that the Steiner trees are pair-wise non-crossing and the sum of their lengths is minimised. 

Recently, Bereg et al.~\cite{bereg2015colored} introduced a $(k+\epsilon)$-approximation for $k$-CESF, which extends on the $k\rho$-approximation (where $\rho$ is the Steiner ratio for which $\rho\leq \frac{2}{\sqrt{3}}$ is its best known upper bound) by Efrate et al.~\cite{efrat2014mapsets}. The $(k+\epsilon)$-approximation computes for points of same colour a $(1+\epsilon)$-approximation of their ESMT. Processing from shortest to longest Steiner tree, the currently processed tree is drawn such that if it intersects any previously processed tree, it is wrapped around it entirely. The result is then pruned and shortened to obtain feasible trees.

Even more recently, Baligács et al.~\cite{baliga} introduced a $(\frac{5}{3}+\epsilon)$-approximation for $k$-ETSP with $k=3$. They adapt Arora's idea to obtain their algorithm.

To our knowledge, no approximation algorithm for $k$-ETSP has been published. 
However, we note that the $(k+\epsilon)$-approximation for $k$-CESF gives rise to a $(2k+\epsilon')$-approximation for $k$-ETSP: 
Double all edges and compute for each colour the DFS-walk/ Eulerian circuit of its tree.

Our initial goal was to find a variegated definition of $k$-CESF or cases in which $k$-CESF admits a better approximation factor with respect to use-cases in cartography. We then shifted to adapting Christofides' approach to $k$-ETSP.
Our work culminated in a $(k+\epsilon)$-approximation which, to our knowledge, is the first approximation algorithm for $k$-ETSP for general $k$.

\section{Notation and definitions}
A \emph{multi-graph} $G$ is a three tuple $(V,E,\Psi)$ of a finite set $V$ called its vertices, a set $E$ called its edges and a function $\Psi: E \to \binom{V}{2}$. We call $\Psi(e)$ the endpoints of an edge $e\in E(G)$.
If two edges $e,e'$ share the same endpoints $\Psi(e)=\Psi(e')$, we say that $e$ and $e'$ are \emph{parallel}. 
While usually a graph is assumed to have no parallel edges,
we here use the term graph also for multi-graphs.

A \emph{geometric graph} $G$ is a graph that comes along with an embedding $\Gamma$ that maps each edge to a curve in the plane with endpoints $\Psi(e)$ and
whose vertices $V(G)\subset \rtwo$ form a planar point set. 
If the edges of $G$ are line segments, we call $G$ \emph{straight-edged}. If not stated otherwise, all geometric graphs will be assumed to be straight-edged.

The set of \emph{incident edges} of a vertex $v$ with respect to $G$ is written as $\delta_G(v)$ and its \emph{degree} is $|\delta_G(v)|$. If it is clear from the context, $G$ is omitted, i.e.\ $\delta(v)$.

Let $G$ be a geometric graph. Define for each $v\in V(G)$ the \emph{ordering} $\omega_v$ as a function that maps a pair $e,e'\in \delta(v)$ to the sequence $\omega_v(e,e')=\sequence{e,...,e'}$ of edges encountered when walking counter-clockwise around $v$ starting from $e$ and ending at $e'$. The function $\omega_v$ corresponds to a directed (counter-clockwise) cyclic order of $\delta(v)$. In case of parallel edges, that is ties in positions, we consider the cyclic order arbitrary but fixed within the following restriction: 
For all parallel edges $e_1,...,e_k$ with endpoints $v$ and $w$, it holds that if $\omega_v(e_1,e_k)=\sequence{e_1,...,e_k}$ then $\omega_w(e_k,e_1)=\sequence{e_k,...,e_1}$.
Therefore, $| \omega_v(e,e')|-1$ is the number of counter-clockwise steps (not including $e$) around $v$ starting from $e$ until $e'$ is first met.

Given a  straight-edged geometric graph $G$, two edges $e,f\in E(G)$ are called \emph{crossing} if they are not parallel and $(\Gamma(e)\cap \Gamma(f))\setminus (\Psi(e) \cup \Psi(f)) \neq \emptyset$. In words, they intersect in their interiors. The graph $G$ is then said to be non-crossing if its edge set is non-crossing.

Given a geometric graph $G$, the length $w(e)$ of an edge $e$, with $w\colon E(G)\to \mathbb{R}_{\geq 0}$, is the (Euclidean) length of its embedding $\Gamma(e)$. If $G$ is straight-edged, this coincides with the Euclidean distance of its endpoints. We define for abbreviation purposes $w(F)=\sum_{f\in F} w(f)$ for $F\subseteq E(G)$. The terms length and weight are used synonymously. 

We define a \emph{walk} $W$ as an alternating sequence $\sequence{v_0,e_1,v_1,e_2,...,e_k,v_k}$ of vertices and edges with $\Psi(e_i)=\set{v_{i-1},v_i}$.
Define $V(W):= \set{v_0,...,v_k}$ and $E(W):=\set{e_1,...,e_k}$. We refer to a walk as \emph{closed} if $v_0=v_k$. Commonly, a walk where $v_0,..., v_{k-1}$ are distinct is called a \emph{path}. A \emph{cycle} is a closed path.

A walk in a geometric graph is self-crossing or just crossing if either two edges $e,f\in E(W)$ cross 
or there exist two pairs of consecutive edges $e_i,e_{i+1}\in E(W)$ and $e_j,e_{j+1}\in E(W)$ incident to a common vertex $v$ such that $|\omega_v(e_i,e_j)|<|\omega_v(e_i,e_{i+1})|< |\omega_v(e_i,e_{j+1})|$ or $|\omega_v(e_i,e_{j+1})|<|\omega_v(e_i,e_{i+1})|< |\omega_v(e_i,e_{j})|$. We refer to this case as \emph{inner-node-intersection}.

Two edge-disjoint walks $W$ and $W'$ in a geometric graph $G$ are \emph{crossing} if either: 1.) there exist edges $e\in E(W)$ and $e'\in E(W')$ such that $e$ and $e'$ cross, or 2.) there exist consecutive edge-pairs $e_i,e_{i+1}\in E(W)$ and $e_j,e_{j+1}\in E(W')$ incident to a common vertex $v$ such that $|\omega_v(e_i,e_j)|<|\omega_v(e_i,e_{i+1})|< |\omega_v(e_i,e_{j+1})|$ or $|\omega_v(e_i,e_{j+1})|<|\omega_v(e_i,e_{i+1})|< |\omega_v(e_i,e_{j})|$.
A set of walks $\mathcal{W}$ is non-crossing if each $W\in \mathcal{W}$ is non-crossing and the walks in $\mathcal{W}$ are pair-wise non-crossing.
Let $G$ be a straight-edged graph. Its planarized graph $G\planarized$ is defined as $V(G\planarized):= V(G)\cup \setcond{p\in \rtwo}{\exists e,f\in E(G): \set{p}=\Gamma(e)\cap \Gamma(f)}$. For each edge $e'\in E(G)$, we add for each pair $p,q\in \Gamma(e')\cap V(G \planarized)$ an edge $e$ if $\forall z \in (\Gamma(e)\cap V(G \planarized)\setminus \set{p,q})\colon z\notin [p,q]$. We set $\Psi(e)=\set{p,q}$ and $\Gamma(e)=[p,q]$ for $G\planarized$. In words, we add a vertex for each intersection point between edges and split the edges accordingly. 

A graph $G$ is Eulerian if it is connected and there exists a closed walk in $G$ that visits each of its edges exactly once. This is equivalent to $G$ being connected and each of its vertices having an even degree \cite{korte2011graphs}. We call such walks Eulerian circuits.

A finite set of points $P\subset \rtwo$ in the plane is in \emph{general position}, if no three points of $P$ lie on a common line $\forall a \neq b\neq c\in P: c\notin \setcond{\lambda a +(1-\lambda)b }{\lambda \in \mathbb{R}}$.

A \emph{curve} is a continuous function $\varphi\colon [0,1]\to \rtwo $. It is called closed if $\varphi(0)=\varphi(1)$. A curve $\varphi$ is called (self)-intersecting if it is not injective on $(0,1)$ and $\varphi(0),\varphi(1)\notin \varphi((0,1))$.
Two curves $\varphi$ and $\gamma$ are said to intersect if their images intersect.
We define a \emph{tour} $T$ of points $P\subset \rtwo$ as a closed curve visiting all points in $P$.

A curve $\varphi$ is \emph{polygonal} if its image is partition-able into a finite set of line segments. 
Then $\varphi$ can be expressed as a walk in a straight-edged geometric graph.
A polygonal curve is \emph{non-crossing} if there exists a straight-edged geometric graph $G$ in which $\varphi$ corresponds to a non-crossing walk in $G$ that visits all its edges.
Using this notion, we can set two polygonal curves $\varphi$ and $\gamma$ in relation by defining their \emph{union} as the planarization $U\planarized$ of the graph $U$ with $V(U)=V(\varphi)\cup V(\gamma)$ and $E(U)=E(\varphi)\cup E(\gamma)$.
Then, two polygonal curves $\varphi$ and $\gamma$ are said to be non-crossing if there exist corresponding non-crossing walks in their union.
Analogously to a set of walks, a set of polygonal curves is defined to be non-crossing, if each curve is non-crossing and each pair of curves is non-crossing.

Let $\varphi$ be a tour of the points $a,b,c$ and $\varphi'$ be a closed curve that encloses the points $a,b,c$ and does not intersect $\varphi$.
If we were to minimise their total length whilst maintaining their topology, one approaches but never reaches the total length of two non-crossing polygonal tours; see \autoref{fig:limit-curve}. Therefore, if we define $k$-ETSP based on non-intersecting curves, its optimal value would not always be properly defined.

Vice versa, given a set of non-crossing polygonal tours, we can construct non-intersecting curves by drawing the polygonal curves with margin to one another whilst putting up with increased total length; see \autoref{fig:walks-to-curves}. 

\begin{figure}[t]
    \begin{subfigure}{0.48\textwidth}
        \centering
        \includegraphics[width=0.5\linewidth]{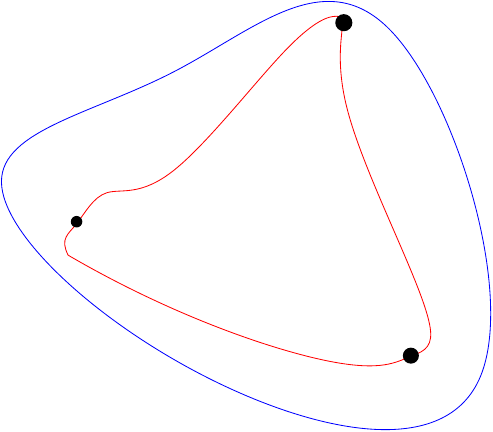}
        \caption{Two non-crossing curves: Red curve is a tour of three points; Blue curve is enclosing the red curve}
        \label{fig:limit-curve-a}
    \end{subfigure}
    \begin{subfigure}{0.48\textwidth}
        \centering
        \includegraphics[width=0.5\linewidth]{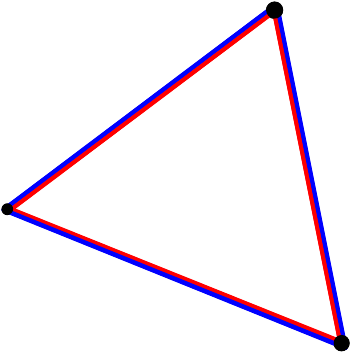}
        \caption{Limes of curves in graph.}
        \label{fig:limit-curve-b}
    \end{subfigure}
    
    \caption{The limit of closed curves can be depicted as closed walks in a graph.}
    \label{fig:limit-curve}
\end{figure}

\section{$(k+\epsilon)$-approximation for $k$-ETSP}
We now state our main problem.
\begin{problem}\label{prob:kctsp}
  \problemtitle{\textbf{$k$-coloured Non-crossing Euclidean Travelling Salesperson Problem ($k$-ETSP)}}
  \probleminput{Set of $k$ finite disjoint point sets in the plane $V_1,...,V_k$.}
  \problemquestion{Non-crossing set of polygonal tours $T_i$ for $V_i$ for all $1\leq i \leq k$ of minimum total length.}
\end{problem}
We state our algorithm:\bigskip\newline

\underline{\textbf{$(k+\epsilon)$-approximation for $k$-ETSP}}:
\begin{compactitem}[leftmargin=5pt]
    \item[Step 1:] For each $V_i$, compute a $(1+\frac{\epsilon}{k})$-approximation $T_i'$ for Euclidean TSP.
    \item[Step 2:] Define $G$ as the union of all $T_i'$. Planarize $G$ to $G\planarized$ and compute a non-crossing Eulerian circuit for each connected component of $G\planarized$.
    \item[Step 3:] 
    For each Eulerian circuit, create a
    copy for each colour in the respective connected component, as described in \autoref{lem:k-coloured-geo-euler}.
\end{compactitem}
\hspace{1.5cm}

In Step~1, we use Arora's PTAS for ETSP, for its approximation guarantee.
In Step~2, we apply a specialisation of Hierholzer's algorithm (see \cite{korte2011graphs}). The original algorithm computes an Eulerain circuit for an Eulerian graph $G$ in $\mathcal{O}(|E(G)|)$ time. 
Below, in \autoref{alg:geo-euler}, we modify it to obtain a non-crossing Eulerian circuit for a non-crossing Eulerian graph. In Step 3, we construct the final tours according to the constructive proof of \autoref{lem:k-coloured-geo-euler}. 
We note that one only has to compute the tours for the colours appearing in the respective connected component.
In a post-process, one can prune and shorten the tours as desired, for example, to avoid visiting vertices introduced by planarization. 

The algorithm's run-time is $\mathcal{O}(kn^ 3 (\log n)^{\mathcal{O}(\epsilon^{-1})})$, where $n:= \sum_{i=1}^k |V_i|$, as it is dominated by the use of the Arora's PTAS in Step~1, which runs in $\mathcal{O}(|V_i|^ 3 (\log |V_i|)^{\mathcal{O}(\epsilon^{-1}))} )$ for each $V_i$.

\begin{figure}[tb]
    \centering
    \begin{subfigure}{0.48\textwidth}
        \centering
        \includegraphics[width=0.6\linewidth]{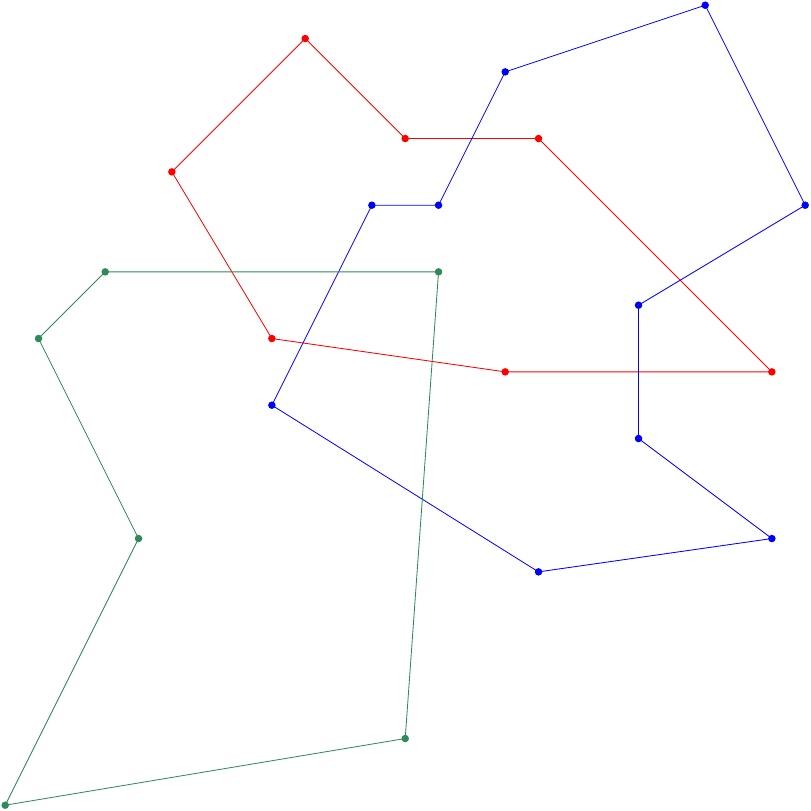}
        \caption{Union of ETSP approximations.}
        \label{fig:tsp-concat}
    \end{subfigure}
    \hfill
    \begin{subfigure}{0.48\textwidth}
        \centering
        \includegraphics[width=0.6\linewidth]{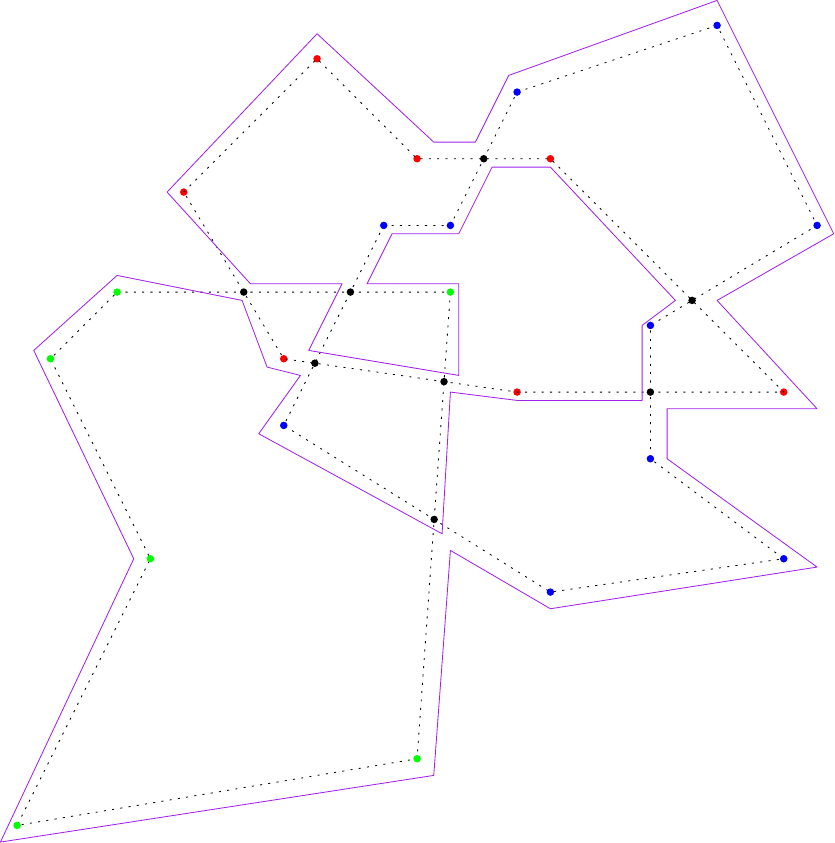}
        \caption{Non-crossing Eulerian circuit of planarized graph.}
        \label{fig:non-x-euler-graph}
    \end{subfigure}
    \caption{Illustration of the first two construction steps.}
    \label{fig:construction}
\end{figure}

\begin{algorithm}
\caption{Non-crossing Hierholzer }\label{alg:geo-euler}
\begin{algorithmic}[1]
    \State initialize $W$ as empty walk
    \While{$E(G)\setminus E(W) \neq \emptyset$}
    \If{ $E(W)=\emptyset$} 
        \State pick any $e\in E(G)$
        \State set $x,y$ equal to endpoints of $e$  
    \Else 
        \State choose $e\in E(G)\setminus E(W)$ with $\Psi(e)=\set{v,w}$ and $v$ is visited by $W$
        \State set $x=v$ and $y=w$
    \EndIf
    \State initialize $C=x,e,w$
    \While{ $x\neq y$}
        \State pick edge $f= \argmin\setcond{|\omega_y(e,f')|}{f'\in (\delta(y)\setminus E(W) )}$
        \State add $f$ to $C$
        \State set $y$ as the other endpoint of $f$
    \EndWhile
    \State let $e,e'$ with $W= \sequence{...,e,x,e',...}$ 
    \State substitute $x$ with $C$, set $W=\sequence{...,e,C,e',...}$
    \EndWhile
    \Return $W$
\end{algorithmic}
\end{algorithm}

\begin{figure}[tb]
    \centering
    \includegraphics[width=0.5\linewidth]{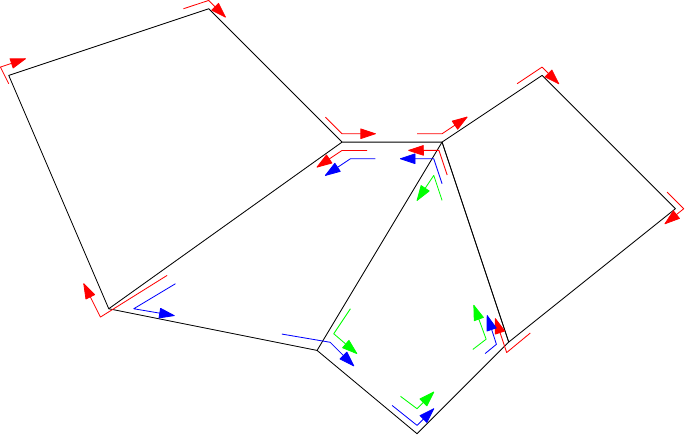}
    \caption{Non-crossing polygonal curves}
    \label{fig:placeholder}
\end{figure}

\begin{lemma}\label{lem:geo-euler}
 Let $G$ be a non-crossing Eulerian geometric graph. If for each vertex $v$ and $e\in\delta(v)$ we can retrieve $e$'s counter-clockwise consecutive edge in $\mathcal{O}(1)$, then we can construct a non-crossing Eulerian circuit $W$ of $G$ in $\mathcal{O}(|E(G)|)$.   
\end{lemma}
\begin{proof}
    As \autoref{alg:geo-euler} is a specialisation of Hierholzer's algorithm, its output is an Eulerian circuit.
    \newline $W$ is non-crossing: \newline
    As $G$ is non-crossing, no two edges cross. Therefore, we only have to consider inner-node-intersections.
    Consider Line 15. If $C$ and $W$ are non-crossing to themselves and each other, then inserting $C$ into $W$, as described, results in a non-crossing walk.
    We will now show that Line 11, during $C$'s construction, will not cause a crossing.
    Assume an inner-node-crossing at $v$ by a quadruple $e_{i},e_{i+1},e_j,e_{j+1}\in \delta(v)$ where the edges' indices $i,j$ coincide with their position in $W$ and $i+1<j$. 
    Assume that the pair $e_i,e_{i+1}$ was appended to $W$ or $C$ before $e_j,e_{j+1}$. 
    Arriving at $v$ over $e_j$, we consider the following two cases: \newline
    \textbf{Case 1.)}: $|\omega_v(e_i,e_j)|<|\omega_j(e_i,e_{i+1})|$\newline
    This is a contradiction to the choice of $e_{i+1}$.\newline
    \textbf{Case 2.)}: $|\omega_v(e_i,e_{j+1})|<|\omega_j(e_i,e_{i+1})|$\newline    
    This is also a contradiction to the choice of $e_{i+1}$. 
\end{proof}
The lemma above seems to be a known result even though we were not able to find a written proof of the statement, which is why we included it here.

\begin{prop}
Let $G$ be a geometric graph. Then
\begin{compactitem}
    \item[i.)] $w(E(G\planarized))=w(E(G))$.
    \item[ii.)] If $V(G)$ is in general position then $G\planarized$ has the same odd vertices as $G$.
    \item[iii.)] If $G$ has no odd vertices then $G\planarized$ also has no odd vertices. 
\end{compactitem}
\end{prop}
\begin{proof}
     Any new vertex $v\in V(G\planarized)\setminus V(G)$ is the intersection of $k$ edges in $E(G)$. After the 'split' those edges result in $2k$ $v$-incident edges. 
     If $V(G)$ is not in general position, an intersection point $p$ of $k$ edges can be equal to a vertex $v$ in $V(G)$. 
     Then $|\delta_{G\planarized}(v)|=|\delta_G(v)|+2k$.
\end{proof}

\begin{lemma}\label{lem:k-coloured-geo-euler}
    Let $G$ be a $k$-coloured non-crossing Eulerian geometric graph. 
    Then we can construct $k$ non-crossing polygonal tours, one corresponding to each colour, each of length $w(E(G))$.
\end{lemma}
\begin{proof}
    We first construct each tour, then proof that they form a non-crossing set.
     Let $W$ be the non-crossing Eulerian circuit given by \autoref{alg:geo-euler}. Define $\overline{G}$ with $V(\overline{G})=V(G)$.
     For each $e\in E(G)$ define $\copies(e)=\set{\overline{e}_1,...,\overline{e}_k}$ with $\overline{\Gamma}(\overline{e}_i)=\Gamma(e)$, $\overline{\Psi}(\overline{e}_i)=\Psi(e)$.
    Add $ \copies(e)$ to $E(\overline{G})$ for each $e\in E(G)$.     
     We maintain any cyclic order $\omega_v(f,g)=\sequence{f,...,e,...,g}$ of $\delta_G(v)$ (for some $f,g\in E(G)$) in $\overline{G}$ by substituting each $e$ with its $\copies(e)=\set{\overline{e}_1,...,\overline{e}_k}$ in place i.e. $\overline{\omega}_v(\overline{f}_i,\overline{g}_j)=\sequence{\overline{f}_i,...,\overline{e}_1,...,\overline{e}_k,...,\overline{g}_j}$ with $\overline{f}_i\in \copies(f)$ and $\overline{g}_j\in \copies(g)$.
    
     Initialize empty walks $T_1,...,T_k$. For each $e_j\in E(W)$, where $W=\sequence{...,v_{j-1},e_j,v_j,...}$, we proceed as follows for $j=1$ to $k$.
     Let $\copies(e_j)=\set{\overline{e}_1,...,\overline{e}_k}$. Assume without loss of generality that their indices are consistent with the cyclic order of $\delta(v_{j-1})$, that is
     $\omega_{v_{j-1}}(\overline{e}_1,\overline{e}_k)=\sequence{\overline{e}_1,...,\overline{e}_k}$.
     Add $\overline{e}_i$ to $T_i$ for each $1\leq i \leq k$.
     
     We now prove that $T_1,...,T_k$ are pair-wise non-crossing. Each tour is non-crossing as $W$ is non-crossing.
     We show that each $v_l\in V(\overline{G})$ does not admit an inner-node-intersection for any $T$ and $T'$.
     Let $\overline{e},\overline{f}\in E(T)$ with $T=\sequence{...,\overline{e},v_l,\overline{f},...}$ and $\overline{e}',\overline{f}'\in E(T')$ with $T'=\sequence{...,\overline{e}',v_l,\overline{f}',...}$.

     Assume $\overline{e},\overline{e}',\overline{f},\overline{f}'$ arise as copies from four distinct edges in $E(G)$. If they form an inner-node-intersection at $v_l$ then $W$ would admit an inner-node-intersection as well.

     As $\overline{e},\overline{f}$ are not copies of the same edge (same holds for $\overline{e}',\overline{f}'$), we are left to consider the case $\overline{e},\overline{e}'\in \copies(e)$ and $\overline{f},\overline{f}'\in \copies(f)$ for some $e,f \in E(G)$. Consider $e$ and $f$ fixed from now.
     We now have $T=\sequence{...,v_{l-1},\overline{e},v_l,\overline{f},v_{l+1},...}$ and $T'=\sequence{...,v_{l-1},\overline{e}',v_l,\overline{f}',v_{l+1},...}$
     
     Consider $\copies(f)=\set{\overline{f}_1,...,\overline{f}_k}$ and $\copies(e)=\set{\overline{e}_1,...,\overline{e}_k}$.
     Without loss of generality, assume $\omega_{v_{l-1}}(\overline{e}_1,\overline{e}_k)=\sequence{\overline{e}_1,...,\overline{e}_k}$ and $\omega_{v_{l}}(\overline{f}_1,\overline{f}_k)=\sequence{\overline{f}_1,...,\overline{f}_k}$.  By the consistency between $\omega_{v_{l-1}}$ and $\omega_{v_{l}}$, it follows that $\omega_{v_{l}}(\overline{e}_k,\overline{e}_1)=\sequence{\overline{e}_k,...,\overline{e}_1}$.
     Assume $\overline{e}=\overline{e}_i$ and $\overline{e}'=\overline{e}_j$. Then by construction $\overline{f}=\overline{f}_i$ and $\overline{f}'=\overline{f}_j$, in particular $T=T_i$ and $T'=T_j$.
     \newline
     \textbf{Case 1.)}: $i<j$ \newline
        Then  $\omega_{v_l}(\overline{e},\overline{e}')=\sequence{\overline{e},...,\overline{f},...,\overline{f}',...,\overline{e}'}=
        \sequence{\overline{e}_i,....\overline{f}_i,...,\overline{f}_j,...,\overline{e}_j}$.\newline
      \textbf{Case 2.)}: $j<i$ \newline
     Then $\omega_{v_l}(\overline{e},\overline{e}')=\sequence{\overline{e},...,\overline{e}'}=\sequence{\overline{e}_i,...,\overline{e}_j}$ or equivalently formulated
     $\omega_{v_l}(\overline{e}',\overline{e})=\sequence{\overline{e}_j,...,\overline{f}_j,...,\overline{f}_i,...,\overline{e}_i}$.
\end{proof}

As mentioned before, we can obtain non-intersecting tours from our non-crossing polygonal tours.
Processing along $W$, draw all polygonal tours in parallel and slightly margined to each other where $\varphi_i$ is enclosed by $\varphi_j$ if $j>i$; see the construction in the proof of \autoref{lem:k-coloured-geo-euler}.
Around a vertex $v$ of colour $i$, if it is first met by $W$, draw $\varphi_i$ through it. If $j>i$ then draw $\varphi_j$ 'above' $v$ else draw it 'below' $v$.
If $v$ is encountered again, draw the connection of each curve $\varphi_1,...,\varphi_k$ to the next edge around $v$; see \autoref{fig:circuit-to-curve} and \ref{fig:copies-to-curves}.

\begin{figure}[tb]
     \centering
     \begin{subfigure}{0.48\textwidth}
        \centering
        \includegraphics[width=\linewidth]{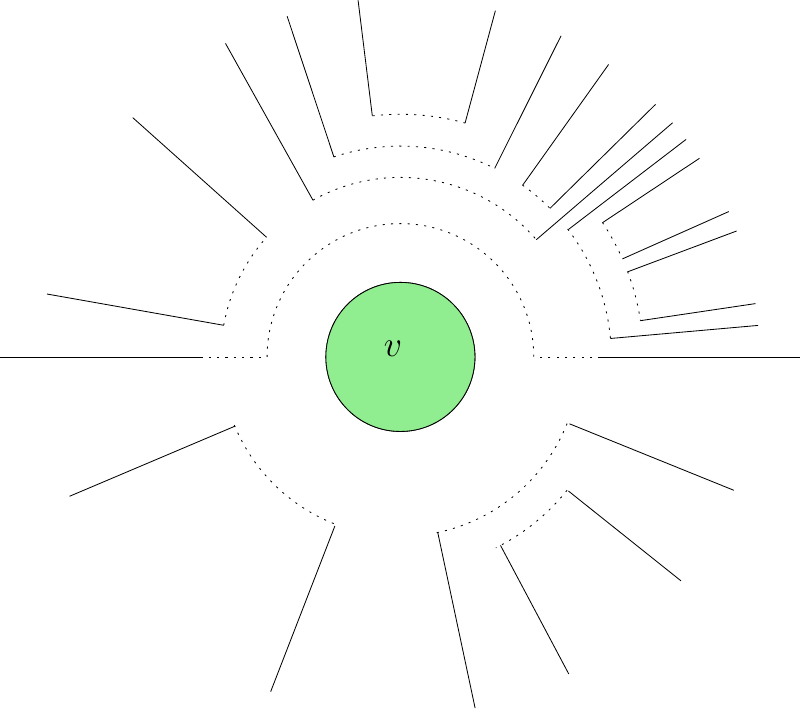}
        \caption{ We transform an non-crossing Eulerian circuit to a non-intersecting closed curve.
        Consecutive edge pairs in the circuit are connected via dotted lines. }
     \label{fig:circuit-to-curve}
     \end{subfigure}
     \begin{subfigure}{0.48\textwidth}
            \includegraphics[width=\linewidth]{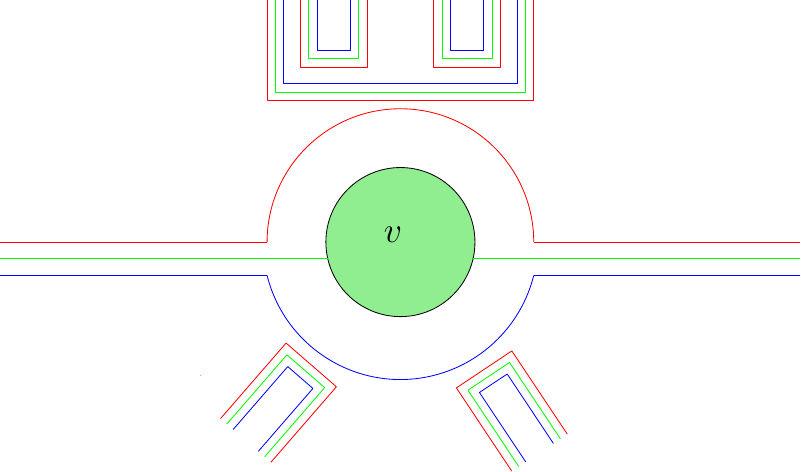}
            \caption{How one can retrieve non-intersecting tours from pair-wise non-crossing polygonal tours}
                 \label{fig:copies-to-curves}
     \end{subfigure}
     \caption{From non-crossing walks to non-intersecting curves.}
     \label{fig:walks-to-curves}
 \end{figure}

Let $k\shortminus \etsp(V_1,...,V_k)$ denote an instance of $k$-ETSP for disjoint point sets $V_1,...,V_k$, analogously, $\etsp(V)$ is defined for a point set $V$. Let $\opt(\mathcal{I})$ denote the value of an optimal solution of a problem instance $\mathcal{I}$.
\setcounter{theorem}{0}
\begin{thm}
    The algorithm's output $T_1,...,T_k$ is a $(k+\epsilon)$-approximation of $\opt(k\shortminus \etsp(V_1,...,V_k))$.
\end{thm}
\begin{proof}
    We have $\opt(k\shortminus\etsp(V_1,...,V_k)) \geq \sum_{i=1}^k \opt(\etsp(V_i))$. It follows 
    \begin{align*}
        \sum_{i=1}^{k} w(E(T_i)) & \leq k(1+\frac{\epsilon}{k})(\sum_{i=1}^k\opt(\etsp(V_i))) \\
            & \leq (k+\epsilon)\opt(k\shortminus\etsp(V_1,...,V_k))
    \end{align*}
\end{proof}
\section{Acknowledgement}
The authors thank Marena Richter for her helpful feedback.
\bibliographystyle{plainurl}
\bibliography{ref}

@inproceedings{baliga,
  author =	{Balig\'{a}cs, J\'{u}lia and Disser, Yann and Feldmann, Andreas Emil and Zych-Pawlewicz, Anna},
  title =	{{A ($5/3+\epsilon$)-Approximation for Tricolored Non-Crossing Euclidean TSP}},
  booktitle =	{32nd Annual European Symposium on Algorithms (ESA 2024)},
  pages =	{15:1--15:15},
  series ={Leibniz International Proceedings in Informatics (LIPIcs)},
  ISBN =	{978-3-95977-338-6},
  ISSN =	{1868-8969},
  year =	{2024},
  volume =	{308},
  editor =	{Chan, Timothy and Fischer, Johannes and Iacono, John and Herman, Grzegorz},
  publisher =	{Schloss Dagstuhl -- Leibniz-Zentrum f{\"u}r Informatik},
  address =	{Dagstuhl, Germany},
  URL =		{https://drops.dagstuhl.de/entities/document/10.4230/LIPIcs.ESA.2024.15},
  URN =		{urn:nbn:de:0030-drops-210862},
  doi =		{10.4230/LIPIcs.ESA.2024.15},
    eprint= {2402.13938},
    eprinttype  = {arxiv}
}

@article{arora1998polynomial,
  title={Polynomial time approximation schemes for {E}uclidean traveling salesman and other geometric problems},
  author={Arora, Sanjeev},
  journal={Journal of the ACM (JACM)},
  volume={45},
  number={5},
  pages={753--782},
  year={1998},
  publisher={ACM New York, NY, USA}
}

@inproceedings{bereg2015colored,
  title={Colored non-crossing {E}uclidean {S}teiner forest},
  author={Bereg, Sergey and Fleszar, Krzysztof and Kindermann, Philipp and Pupyrev, Sergey and Spoerhase, Joachim and Wolff, Alexander},
  booktitle={International Symposium on Algorithms and Computation},
  pages={429--441},
  year={2015},
  organization={Springer}
}

@incollection{korte2011traveling,
  title={The {T}raveling {S}alesman {P}roblem},
  author={Korte, Bernhard and Vygen, Jens},
  booktitle={Combinatorial Optimization: Theory and Algorithms},
  pages={557--592},
  year={2011},
  publisher={Springer}
}

@inproceedings{efrat2014mapsets,
  title={MapSets: Visualizing embedded and clustered graphs},
  author={Efrat, Alon and Hu, Yifan and Kobourov, Stephen G and Pupyrev, Sergey},
  booktitle={International Symposium on Graph Drawing},
  pages={452--463},
  year={2014},
  organization={Springer}
}

@incollection{korte2011graphs,
  title={Graphs},
  author={Korte, Bernhard and Vygen, Jens},
  booktitle={Combinatorial Optimization: Theory and Algorithms},
  pages={15--54},
  year={2011},
  publisher={Springer}
}

@article{VANBEVERN2020118,
    title = {A historical note on the 3/2-approximation algorithm for the metric traveling salesman problem},
    journal = {Historia Mathematica},
    volume = {53},
    pages = {118-127},
    year = {2020},
    issn = {0315-0860},
    doi = {https://doi.org/10.1016/j.hm.2020.04.003},
    url = {https://www.sciencedirect.com/science/article/pii/S0315086020300240},
    author = {René {van Bevern} and Viktoriia A. Slugina}
}

@article{van2024simplesets,
  title={SimpleSets: Capturing Categorical Point Patterns with Simple Shapes},
  author={van den Broek, Steven and Meulemans, Wouter and Speckmann, Bettina},
  journal={IEEE Transactions on Visualization and Computer Graphics},
  volume={31},
  number={1},
  pages={262--271},
  year={2024},
  publisher={IEEE}
}

\end{document}